\documentclass[sigconf]{acmart}
\renewcommand\footnotetextcopyrightpermission[1]{} 

\usepackage{amsthm,amsmath,amssymb}
\usepackage{customcommands}
\usepackage{enumitem}
\usepackage{url}
\usepackage{color}
\usepackage{wrapfig}
\usepackage{xspace}
\usepackage{subfig}
\usepackage{caption}
\usepackage{algorithm}
\usepackage[noend]{algpseudocode}

\usepackage{tikz}

\renewcommand{\hat}{\widehat} 

\newcommand{\ED}{\textnormal{ED}}

\newcommand{\calH}{\mathcal{H}}
\newcommand{\calT}{\mathcal{T}}

\newcommand{\defn}{\emph}

\newcommand{\poly}{\textnormal{poly}}

\let\oldtexttt\texttt
\renewcommand{\texttt}[1]{\xspace{\normalfont{\oldtexttt{#1}}}\xspace}

\newcommand{\ins}{\text{\texttt{insert}}}
\newcommand{\del}{\text{\texttt{delete}}}
\newcommand{\rep}{\text{\texttt{replace}}}
\newcommand{\loo}{\text{\texttt{loop}}}
\newcommand{\match}{\text{\texttt{match}}}
\newcommand{\die}{\text{\texttt{stop}}}

\title{Approximate Similarity Search Under Edit Distance Using Locality-Sensitive Hashing}
\author{Samuel McCauley}
\email{sam@cs.williams.edu}
\affiliation{
  \institution{Williams College}
  \city{Williamstown, MA 01267 USA}
}
\date{}

\begin{document}
\sloppy

\begin{abstract}
  Edit distance similarity search, also called approximate pattern matching, is a fundamental problem with widespread database applications.  The goal of the problem is to preprocess $n$ strings of length $d$, to quickly answer queries $q$ of the form: if there is a database string within edit distance $r$ of $q$, return a database string within edit distance $cr$ of $q$. 
  Previous approaches to this problem either rely on very large (superconstant) approximation ratios $c$, or very small search radii $r$.  Outside of a narrow parameter range, these solutions are not competitive with trivially searching through all $n$ strings.

  In this work give a simple and easy-to-implement hash function that can quickly answer queries for a wide range of parameters.   
  Specifically, our strategy can answer queries in time $\tilde{O}(d3^rn^{1/c})$.  
  The best known practical results require $c \gg r$ to achieve any correctness guarantee; meanwhile, the best known theoretical results are very involved and difficult to implement, and require query time at least $24^r$.  Our results significantly broaden the range of parameters for which we can achieve nontrivial bounds, while retaining the practicality of a locality-sensitive hash function.

  We also show how to apply our ideas to the closely-related Approximate Nearest Neighbor problem for edit distance, obtaining similar time bounds. 


\end{abstract}

\maketitle


\section{Introduction}
\seclabel{intro}

For a large database of items, a \defn{similarity search} query asks which database item is most similar to the query.  This leads to a basic algorithmic question: how can we preprocess the database to answer these queries as quickly as possible?

Similarity search is used frequently in a wide variety of applications.  Unfortunately, for databases containing high-dimensional items, algorithm designers have had trouble obtaining bounds that are significantly faster than a linear scan of the entire database.  This has often been referred to as the ``curse of dimensionality.''  Recent work in fine-grained complexity has begun to explain this difficulty: achieving significantly better than linear search time would violate the strong exponential time hypothesis~\cite{AlmanWilliams15,Rubinstein18,CohenAddadFeSt19}.

However, these queries can be relaxed to \defn{approximate similarity search} queries.  For an approximation factor $c$, we want to find a database item that is at most a $c$ factor less similar than the most similar item.  

Approximate similarity search is fairly well-understood for many metrics; see~\cite{AndoniInRa18} for a survey.  For example, in Euclidean space we have theoretical upper bounds~\cite{AndoniLaRa17,Christiani17}, fast implementations~\cite{JohnsonDoJe17,MalkovYashunin18,AumullerBeFa19,DongMoLi11,LiZhYi16,WangShWa18}, and lower bounds for a broad class of algorithms~\cite{AndoniLaRa17}.  
Many of these results are based on \defn{locality-sensitive hashing} (LSH), originally described in~\cite{IndykMotwani98}.  A hash is locality-sensitive if similar items are more likely to share the same hash value.  

When a database contains text items, a natural notion of similarity is edit distance: how many character inserts, deletes, and replacements are required to get from the query string to a database string?  In fact, edit distance similarity search is frequently used in computational biology~\cite{KahveciLjSi04,OzturkFerhatosmanoglu03,LamSuTa08}, spellcheckers~\cite{WilburKiXi06,BrillMoore00}, computer security (in the context of finding similarity to weak passwords)~\cite{ManberWu94}, and many more applications; see e.g.~\cite{Boytsov11}.

Surprisingly, finding an efficient algorithm for approximate similarity search under edit distance remains essentially open.  Known results focus on methods for exact similarity search (with $c=1$), which incur expensive query times, and on embeddings, which require very large---in fact superconstant---approximation factors $c$.  

However, recent work provides a potential exception to this.  The CGK embedding~\cite{ChakrabortyGoKo16} is simple and practical, and embeds into Hamming space with stretch $O(r)$---in particular, it does well when the distance between the closest strings is fairly small.  EmbedJoin, a recent implementation by Zhang and Zhang~\cite{ZhangZhang17}, showed that the CGK embedding performs very well in practice.  EmbedJoin first embeds each string into Hamming space using the CGK embedding.  Then, the remaining nearest neighbor search\footnote{Zhang and Zhang investigated similarity \emph{joins}, in which all similar pairs in a set are returned, rather than preprocessing for individual nearest neighbor queries.  However, their ideas can be immediately generalized.} is done using the classic bit sampling LSH for Hamming distance.  Each of these steps---both the CGK embedding and the bit sampling LSH---is repeated several times independently.  This method gave orders of magnitude better performance than previous methods.  Furthermore, their results greatly outperformed the worst-case CGK analysis.

Thus, several questions about using CGK for edit distance similarity search remained.  Zhang and Zhang used several CGK embeddings, performing a sequence of Hamming distance hashes for each---can these two steps be combined into a single method to improve performance?  Meanwhile, their tests focused on practical datasets; is it possible to provide worst-case bounds for this method, ensuring good performance for any dataset?

In this paper we answer these questions in the affirmative.
In doing so, we give the first locality-sensitive hash for edit distance with worst-case guarantees.

\subsection{Results}

The main result of our paper is the first locality-sensitive hash for edit distance.  We analyze the performance of this hash when applied to the problems of approximate similarity search and approximate nearest neighbor search, obtaining time bounds that improve on the previously best-known bounds for a wide range of important parameter settings.

Let $n$ be the number of strings stored in the database.  
We assume that all query strings and all database strings have length at most $d$.  We assume $d = O(n)$ and the alphabet size is $O(n)$.\footnote{Usually $d$ and the alphabet size are much smaller.  If this assumption does not hold, it is likely that a completely different approach will be more successful: for example, if $d=\poly(n)$, then the method used to calculate the edit distance between two strings becomes critically important to the query time.}

Our first result analyzes the time and space required by our LSH to solve the approximate similarity search problem.  This data structure works for a fixed radius $r$: for each query, if there exists a database point within distance $r$, we aim to return a database point within distance $cr$.

\begin{theorem}
  \thmlabel{similarity}
  There exists a data structure answering Approximate Similarity Search queries under Edit Distance in $\tilde{O}(d3^rn^{1/c})$ time per query and $\tilde{O}(3^rn^{1 + 1/c} + dn)$ space.
\end{theorem}

 We also give a data structure that answers queries where the distance $r$ to the closest neighbor is not known during preprocessing. We call this the approximate \emph{nearest} neighbor search problem.

\begin{theorem}
  \thmlabel{nearestneighbor}
  There exists a data structure answering Approximate Nearest Neighbor Search queries under Edit Distance in $\tilde{O}(d3^rn^{1/c})$ time per query and $\tilde{O}(n^2)$ space.
\end{theorem}

\paragraph{Implications for Related Problems}
Our results lead to immediate bounds for similarity join, where all close pairs in a database are computed; see e.g.~\cite{PaghPhSi17,ZhangZhang17,ZhangZhang19}.  

Much of the previous work on approximate similarity search under edit distance considered a variant of this problem: there is a long text $T$, and we want to find all locations in $T$ that have low edit distance to the query $q$.  Our results immediately apply to this problem by treating all $d$-length substrings of $T$ as the database of items.

Frequently, practical situations may require that we find all of the neighbors with distance at most $r$, or (similarly) the $k$ closest neighbors. See e.g.~\cite{AhleAuPa17} for a discussion of this problem in the context of LSH.  Our analysis immediately applies to these problems.  However, if there are $k$ desired points, the running time increases by a factor $k$.

\subsection{Comparison to Known Results}
\seclabel{related}

In this section, we give a short summary of some key results for edit distance similarity search.  We focus on algorithms that have worst-case query time guarantees.  We refer the reader to~\cite{ZhangZhang17,ZhangZhang19} as good resources for related practical results, and~\cite{MaassMoNo05,Navarro01,Boytsov11} for a more extensive discussion of related work on the exact problem (with $c=1$).

\paragraph{Exact Similarity Search Under Edit Distance}
Exact similarity search under edit distance (i.e. with $c = 1$) has been studied for many years.
We focus on a breakthrough paper of Cole, Gottlieb, and Lewenstein that achieved space $O(n 5^r (1.5r + \log n)^r/r!)$ and query time $O(d + 6^r(1.5r + \log n)^r/r!)$~\cite{ColeGoLe04}.\footnote{These bounds are a slight simplification of the actual results using the AM-GM inequality.} We will call this structure the CGL tree. These bounds stand in contrast to previous work, which generally had to assume that the length of the strings $d$ or the size of the alphabet $|\Sigma|$ was a constant to achieve similar bounds.   
Later work has improved on this result to give similar query time with linear space~\cite{ChanLaSu06}.

Before comparing to our bounds, let us lower bound the CGL tree query time---while this gives a lower bound on an upper bound (an uncomfortable position since we are not specifying its exact relationship to the data structure), it will be helpful to get a high-level idea of how these results compare.  Using Sterling's approximation, and dropping the $+d$ term, we can simplify the query time to $\tilde{O}((6e/r)^r(1.5r + \log n)^r) \leq \tilde{O}((9e)^r(1 + (\log n)/(1.5r))^r)$.  From this final equation, we can see that even for very small $n$, the guaranteed query time is at least $(9e)^r > 24^r$; if $\log n \gg 1.5r$ it can become much worse.

Comparing the $(9e)^r$ term with our query time of $\tilde{O}(d 3^r n^{1/c})$, it seems that which is better depends highly on the use case---after all, we're exchanging a drastically improved exponential term in $r$ for a polynomial term in $n$. 

However, there is reason to believe that our approach has some significant advantages.  First, for $c$ bounded away from $1$, with moderate $n$ and small $d$, 
the CGL query time rapidly outpaces our own even for small $r$.
Let's do a back-of-the-envelope calculation with some reasonable parameters---we ignore constants here, but note that slight perturbations in $r$ easily make up for such discrepancies. If we have 400k strings of 500 characters\footnote{These are the parameters of the UniRef90 dataset from the UniProt Project~\url{http://www.uniprot.org/}, one protein genome dataset used as an edit distance similarity search benchmark~\cite{ZhangZhang17,ZhangZhang19}; other genomic datasets have (broadly) similar parameters.}  with $c=1.5$, $ 6^r(1.5r + \log n)^r/r! \geq d3^r n^{1/c}$ for $r > 4$.
In other words, even for very small search radii and fairly large $n$ (where the CGL tree excells), the large terms in the base of $r$ can easily overcome a polynomial-in-$n$ term.
Second, the constants in the CGL tree seem to be unfavorable: the CGL tree uses beautiful but nontrivial data structures for LCA and LCP that may add to the constants in the query time.  In other words, it seems likely that the CGL tree is most viable for even smaller values of $r$ than the above analysis would indicate.  

We suspect that these complications are part of the reason why state-of-the-art practical edit distance similarity search methods are based on heuristics or embeddings, rather than tree-based methods (see e.g.~\cite{ZhangZhang19}).




\paragraph{Approximate Similarity Search Under Edit Distance}
Previous results for approximate similarity search with worst-case bounds used either product metrics, or embeddings into $L_1$.

In techniques based on \defn{product metrics}, each point is mapped into several separate metrics.  The distance between two points is defined as their maximum distance in any of these metrics.  Using this concept, Indyk provided an extremely fast (but large) nearest-neighbor data structure requiring $O(d)$ query time and $O(n^{d^{1/(1 + \log c)}})$ space for any $c \geq 3$~\cite{Indyk04}.  

\paragraph{Embedding into $L_1$}
Because there are approximate nearest neighbor data structures for $L_1$ space that require $n^{1/c+ o(1)}$ time and $n^{1 + 1/c+ o(1)}$ space,\footnote{This can be improved to $n^{1/(2c - 1) + o(1)}$ and $n^{1+1/(2c-1)+o(1)}$ time and space respectively using data-dependent techniques, and can be further generalized to other time-space tradeoffs; see~\cite{AndoniLaRa17}.} an embedding with stretch $\alpha$ leads to an approximate nearest neighbor data structure with query time $n^{\alpha/c + o(1)}$ for $c > \alpha$.

A long line of work on improving the stretch of embedding edit distance into $L_1$
ultimately resulted in a deterministic embedding with stretch $\exp(\sqrt{\log d}/\log\log d)$~\cite{OstrovskyRabani07}, and 
a randomized embedding with stretch $O((\log d) 2^{O(\log^* d)} \log^* d)$~\cite{Jowhari12}.  

More recently, the CGK embedding parameterized by $r$ instead of $d$, giving an embedding into Hamming space\footnote{Hamming space and $L_1$ have the same state-of-the-art LSH bounds.}  with stretch $O(r)$~\cite{ChakrabortyGoKo16}.  
However, the constants proven in the CGK result are not very favorable---the upper limit on overall stretch given in the paper is $2592r$ (though this may be improvable with tighter random walk analysis).  
Thus, using the CGK embedding, and then performing the standard bit sampling LSH for Hamming distance on the result, gives an approximate similarity search algorithm with query time $n^{2592r/c + o(1)}$ so long as $c > 2592r$.  We describe in detail how our approach improves on this method in \secref{cgk_comparison}.

Zhang and Zhang~\cite{ZhangZhang17} implemented a modified and improved version of this approach; their results far outperformed the above analysis. Closing this gap between worst-case analysis and practical performance is one contribution of this work.



There is a lower bound of 
$\Omega(\log d)$
for the stretch of any embedding of edit distance into $L_1$%
~\cite{KhotNaor06}.  This implies that embedding into $L_1$ is a hopeless strategy for $c < \log d$, whereas we obtain nontrivial bounds even for constant $c$.
Thus, for this parameter range, using a locality-sensitive hash is fundamentally more powerful than embedding into $L_1$.


\paragraph{Locality-Sensitive Hashing}
An independent construction of an LSH for edit distance was given by Mar{\c{c}}ais et al.~\cite{MarcaisDePa19}.
Their work uses a fundamentally different approach, based on an ordered min-hash of k-mers.  Their results include bounds proving that the hash is locality-sensitive; however, they do not place any worst-case guarantees on the gap between the probability that close points collide and the probability that far points collide ($p_1$ and $p_2$ respectively in \defref{lsh}), and so do not obtain similarity search bounds. 

\paragraph{Exponential search cost}
To our knowledge, a trivial brute force scan is the only algorithm for approximate similarity search under edit distance whose worst-case cost is not exponential in the search radius $r$.  While we significantly improve this exponential term, removing it altogether remains an open problem.  A recent result of Cohen-Addad et al. gave lower bounds showing that, assuming SETH, there exist parameter settings such that cost exponential in $r$ is required for any edit distance similarity search algorithm~\cite{CohenAddadFeSt19}.  Due to some specifics of the parameter settings, their results do not imply that the exponential-in-$r$ term in our query time is necessary; however,  
this may give some indication as to why removing this term has proven so challenging.

\section{Model and Preliminaries}
\seclabel{prelim}


We denote the alphabet used in our problem instance as $\Sigma$.
We use two special characters $\bot$ and $\$$, which we assume are not in $\Sigma$.  
The hash appends $\$ $ to each string being hashed;  we call a string $\$ $-terminal if its last character is $\$ $, and it does not contain $\$$ in any other position.


We index into strings using 0-indexed subscripts; $x_0$ is the first character of $x$ and $x_i$ is the $i+1$st character.  We use $x[i]$ to denote the prefix of $x$ of length $i$; thus $x[i] = x_0\ldots x_{i-1}$.  Finally, we use $x\concat y$ to denote the concatenation of two strings $x$ and $y$, and $|x|$ to denote the length of a string $x$. 

\subsection{Edit Distance}

Edit distance is defined using three operations: inserts, deletes, and replacements.
Given a string $x = x_1x_2\ldots x_d$, inserting a character $\sigma$ at position $i$ results in a string $x' = x_1\ldots x_{i-1} \sigma x_i\ldots x_d$.  
Replacing the character at position $i$ with $\sigma$ results in $x' = x_1\ldots x_{i-1} \sigma x_{i+1}\ldots x_d$.  Finally, deletion of the character at position $i$ results in $x' = x_1\ldots x_{i-1}x_{i+1}\ldots x_d$.  
We refer to these three operations as \defn{edits}.
The edit distance from $x$ to $y$ is defined as the smallest number of edits that must be applied to $x$ to obtain $y$.  We denote this as $\ED(x,y)$.   


\subsection{Model and Problem Definition}

In this paper we solve the approximate similarity search problem under edit distance, which can be defined as follows.

\begin{definition}[Approximate Similarity Search Under Edit Distance]
  \deflabel{ASS}
  Given a set of $n$ strings $S$ and constants $c$ and $r$, preprocess $S$ to quickly answer queries of the form, ``if there exists a $y\in S$ with $\ED(q,y)\leq r$, return a $y'\in S$ with $\ED(q,y')\leq cr$ with probability $> 1/10$.''
\end{definition}

The above is sometimes called the approximate \defn{near} neighbor problem.  The constant $1/10$ is arbitrary and can be increased to any desired constant without affecting our final bounds.  


Oftentimes, we want to find the nearest database item to each query rather than parameterizing explicitly by $r$.

\begin{definition}[Approximate Nearest Neighbor Search Under Edit Distance]
  \deflabel{ANN}
  Given a set of $n$ strings $S$ and a constant $c$, preprocess $S$ to quickly answer queries of the form, ``for the smallest $r$ such that there exists a $y\in S$ with $\ED(q,y)\leq r$, return a $y'\in S$ with $\ED(q,y')\leq cr$ with probability $> 1/10$.''
\end{definition}

For most previous LSH-based approaches, efficient Nearest Neighbor Search algorithms follow immediately from Approximate Similarity Search algorithms using the black box reduction of Har-Peled, Indyk, and Motwani~\cite{HarPeledInMo12}.  However, the exponential dependence on $r$ in our bounds requires us to instead use a problem-specific approach.

\subsection{Locality-Sensitive Hashing}
\seclabel{lsh}

A hash family is \defn{locality sensitive} if close elements are more likely to hash together than far elements. 
Locality-sensitive hashing is one of the most effective methods for approximate similarity search in high dimensions~\cite{HarPeledInMo12,IndykMotwani98,AndoniLaRa17,Charikar02}.

\begin{definition}[Locality-Sensitive Hash]
  \deflabel{lsh}
  A hash family $\calH$ is $(r,cr,p_1,p_2)$-sensitive for a distance function $d(x,y)$ if 
  \begin{itemize}[topsep=0pt,noitemsep] 
    \item for all $x_1$, $y_1$ such that $d(x_1,y_1)\leq r$, $\Pr_{h\in\calH}(h(x_1) = h(y_1)) \geq p_1$, and 
    \item for all $x_2$, $y_2$ such that $d(x_2,y_2)\geq cr$, $\Pr_{h\in\calH}(h(x_1) = h(y_1)) \leq p_2$.
  \end{itemize}
\end{definition}

This paper gives the first direct\footnote{By direct, we mean that this hash does not embed into $L_1$ or use product metrics as an intermediary.} locality-sensitive hash under edit distance with closed forms for $p_1$ and $p_2$.
We bound $p_1$ and $p_2$ for any $r$ and $cr$ in \lemref{closepoints} and \lemref{looseupper} respectively.

Some previous work (i.e.~\cite{ChierichettiKuMa14,ColemanShBa19}) has a stricter definition of locality sensitive hash: it requires that there exists a function $f$ such that $\Pr(h(x) = h(y)) = f(d(x,y))$.  Our hash function does not satisfy this definition; the exact value of $x$ and $y$ is necessary to determine their collision probability (see \lemref{nomatchcollisionprob} for example).  

\paragraph{A Note on Concatenating Hashes}  Most previous approaches to nearest neighbor search begin with an LSH family that has $p_1,p_2 = \Omega(1)$.  A logarithmic number of independent hashes are concatenated together so that the concatenated function has collision probability $1/n$.  This technique was originally developed in~\cite{IndykMotwani98}, and has been used extensively since; e.g. in~\cite{AndoniLaRa17,ChristinianiPagh17,AhleAuPa17}.

However, in this paper, we use a single function each time we hash.  We directly set the hash parameters to achieve a desirable $p_1$ and $p_2$ (in particular, we want $p_2 \approx 1/n$).  This is due to the stray constant term in \lemref{looseupper}.  While our hash could work via concatenating several copies of a relatively large-probability\footnote{Although less than constant---\lemref{closepoints} and the assumption that $p \leq 1/3$ implies $p_1 \leq (1/3)^r$.} LSH, this would result in a data structure with larger space and slower running time.
One interesting implication is that, unlike many previous LSH results, our running time is not best stated with a parameter $\rho = \log p_1 / \log p_2$---rather, we choose our hashing parameters to obtain the $p_1$ and $p_2$ to give the best bounds for a given $r$, $c$, and $n$.

\section{The Locality-Sensitive Hash} \seclabel{algorithm}

Each hash function from our family maps a string $x$ of length $d$ with alphabet $\Sigma$ to a string $h(x)$ with alphabet $\Sigma\cup \{\bot\}$ of length $O(d + \log n)$.  The function scans over $x$ one character at a time, adding characters to $h(x)$ based on the current character of $x$ and the current length of $h(x)$.  Once the function has finished scanning $x$, it stops and outputs $h(x)$.

At a high level, for two strings $x$ and $y$, our hash function can be viewed as randomly guessing a sequence of edits $T$, where $h(x) = h(y)$ if and only if applying the edits in $T$ to $x$ obtains $y$.  Equivalently, one can view the hash as a random walk through the dynamic programming table for edit distance, where matching edges are traversed with probability $1$, and non-matching edges are traversed with a tunable probability $p \leq 1/3$.  We discuss these relationships in \secref{interpreting}. 

Note the contrast with the CGK embedding, which uses a similar mechanism to guess the \emph{alignment} between the two strings for each mismatch, rather than addressing each edit explicitly.  This difference is key to our improved bounds; see \secref{cgk_comparison}.

\paragraph{Parameters of the Hash Function}
We parameterize our algorithm using a parameter $p \leq 1/3$. 
By selecting $p$ we can control the values of $p_1$ and $p_2$ attained by our hash (see Lemmas~\ref{lem:closepoints} and~\ref{lem:looseupper}).  
We will specify $p$ to optimize nearest neighbor search performance for a given $r$, $c$, and $n$ in \secref{settingnn}.
  We split $p$ into two separate parameters $p_a$ and $p_r$ defined as 
  $ p_a = \sqrt{{p}/{1+p}} $ and $ p_r = \sqrt{p}/(\sqrt{1 + p} - \sqrt{p}).  $
Since $p \leq 1/3$, we have $p_a \leq 1/2$ and $p_r \leq 1$. 
For the remainder of this section, we will describe how the algorithm behaves using $p_a$ and $p_r$.
The rationale behind these values for $p_a$ and $p_r$ will become clear in the proof of \lemref{probinduces}. 







\paragraph{Underlying Function}
Each hash function in our hash family has an \defn{underlying function} that maps each \textit{(character, hash position)} pair to a pair of uniform random real numbers:  $\rho: \Sigma\cup\{\$\} \times \{1,\ldots,8d/(1-p_a) + 6\log n\} \rightarrow [0,1) \times [0,1)$.\footnote{Adding $\$$ to the alphabet allows us to hash past the end of a string---this helps with edits that {append} characters.}  We discuss  how to store these functions and relax the assumption that these are real numbers in \secref{practical}.

The only randomness used in our hash function is given by the underlying function.\footnote{In fact, the underlying function is a generalization of the random string used in the CGK embedding.}
In particular, this means that two hash functions $h_1$ and $h_2$ have identical outputs on all strings if their underlying functions $\rho_1$ and $\rho_2$ are identical. 
Thus, we pick a random function from our hash family by sampling a random underlying function. 
We use $h_\rho(x)$ to denote the hash of $x$ using underlying function $\rho$. 


The key idea behind the underlying function is that the random choices made by the hash depend only on the current character seen in the input string, and the current length of the output string.  This means that if two strings are aligned---in particular, if the ``current'' character of $x$ matches the ``current'' character of $y$---the hash of each will make the same random choices, so the hashes will stay the same until there is a mismatch.  This is the ``oblivious synchronization mechanism'' used in the CGK embedding~\cite{ChakrabortyGoKo16}.  

\subsection{How to Hash}
A hash function $h$ is selected from the family $\calH$
 by sampling a random underlying function $\rho$.  We denote the hash of a string $x$ using $\rho$ as $h_{\rho}(x)$. The remainder of this section describes how to determine $h_{\rho}(x)$ for a given $x$ and $\rho$.

To hash $x$, the first step is to append $\$$ to the end of $x$ to obtain $x\concat \$ $.  We will treat $x \gets x \concat \$ $ as the input string from now on---in other words, we assume that $x$ is $\$$-terminal.
Let $i$ be the current index of $x$ being scanned by the hash function. 
We will build up $h_\rho(x)$ character-by-character, storing intermediate values in a string $s$.
The hash begins by setting $i = 0$, and $s$ to the empty string.

The hash function repeats the following process while $i < |x|$ and\footnote{The requirement $|s|< 8d/(1-p_a) +6\log n$ is useful to bound the size of the underlying function in \secref{practical}.  We show in \lemref{shorthash} that this constraint is very rarely violated.} $|s| < 8d/(1-p_a) + 6\log n$. 
The hash first stores the current value of the underlying function based on $x_i$ and $|s|$ by setting $(r_1,r_2)\gets \rho(x_i,|s|)$.  
The hash performs one of three actions based on $r_1$ and $r_2$; in each case one character is appended to the string $s$.  We name these cases a \defn{hash-insert}, \defn{hash-replace}, and \defn{hash-match}.
\begin{itemize}[topsep=1pt,itemsep=1pt]
  \item \textbf{If $r_1 \leq p_a$, hash-insert:} append $\bot$ to $s$.  
  \item \textbf{If $r_1 > p_a$ and $r_2\leq p_r$, hash-replace:}  append $\bot$ to $s$ and increment $i$. 
  \item \textbf{If $r_1 > p_a$ and $r_2 > p_r$, hash-match:} append $x_i$ to $s$ and increment $i$.
\end{itemize}
\vspace{2pt}
When $i \geq |x|$ or $|s| \geq 8d/(1-p_a) + 6\log n$, the hash stops and returns $s$ as $h_\rho(x)$.  

We give pseudocode for this hash function and an example hash in Appendix~\ref{sec:example}.

%
\subsection{Comparing to the CGK Embedding}
\seclabel{cgk_comparison}
Our hash function follows some of the same high-level structure as the CGK embedding~\cite{ChakrabortyGoKo16}. In fact, our hash reduces to their embedding by omitting the appended character $\$$, and setting $p_a = 1/2$ and $p_r = 0$.  

However, our hash has two key differences over simply using the CGK embedding to embed into Hamming space, and then using bit sampling.  These differences work together to allow us to drastically improve the $n^{2592r/c + o(1)}$ bound we obtained in \secref{related}.  

First, we modify $p_a$; that is, we modify the probability that we stay on a single character $x_i$ of the input string for multiple iterations.  Second, we combine the embedding and bit sampling into a single step---this means that we can take the embedding into account when deciding whether to sample a given character.

Combining into one step already gives an inherent improvement.  
After embedding, we do not want to sample a ``repeated'' character---this is far less useful than sampling a character the last time it is written, after the hash has attempted to align them.  Thus, we only sample a character (with probability $1-p_r$) the last time that character is written.  

However, the significant speedup comes from using \emph{repeated} embeddings---in short, at a high level, each LSH in our approach consists of a single CGK embedding with a single bit sampling LSH.  If a single embedding is used, the performance of the algorithm as a whole has the expected stretch of that single embedding as a bottleneck. As a result, the expected stretch winds up in the exponent of $n$, and $c$ must be at least as large as the expected stretch to guarantee correctness.   By repeatedly embedding, our bounds instead depend (in a sense) on the \emph{best}-case stretch over the many embeddings.

These repeated embeddings is where these two differences---modifying $p_a$ and integrating into a single LSH---act in concert.  A back-of-the-envelope calculation implies that a CGK embedding will have stretch\footnote{To be more precise, with probability $1/4^r$ one string with distance $r$ from the query will have embedded Hamming distance $r$, while all strings with distance $x$ will have embedded Hamming distance $\geq x/2$.} 2 with probability at least $1/4^r$.  But this does not immediately imply a good algorithm: if we perform $4^r$ embeddings, how well will we do in the cases that don't have $O(1)$ stretch?  Meanwhile, any constant loss in the analysis winds up in the exponent of $n$---that is to say, the back-of-the-envelope analysis still isn't tight enough.  Overall, with the CGK embedding as a black box, a full analysis would require an analysis (with tight constants) detailing the probability that an embedding has any given stretch.  Instead, by combining these approaches in a single LSH, we can instead model the entire problem as a single random walk in a two-dimensional grid.  

Overall, a combined approach gives better worst-case performance, and a unified (and likely simpler) framework for analysis.

\section{Analysis}
\seclabel{analysis}
In this section we show how analyze the hash given in \secref{algorithm}, ultimately proving Theorems~\ref{thm:similarity} and~\ref{thm:nearestneighbor}.

We begin in \secref{interpreting} with some structure that relates hash collisions between two strings $x$ and $y$ with sequences of edits that transform $x$ into $y$.  We use this to bound the probability that $x$ and $y$ collide in \secref{collisionprobs}.  With this we can prove our main results in Sections~\ref{sec:settingnn} and~\ref{sec:nearestneighbor}. Finally we discuss how to store the underlying functions in \secref{practical}.  

\subsection{Interpreting the Hash}
\seclabel{interpreting} 

In this section, we discuss when two strings $x$ and $y$ hash (with underyling function $\rho$) to the same string $h_\rho(x) = h_\rho(y)$.  

We define three sequences to help us analyze the hash.  In short, the \defn{transcript} of $x$ and $\rho$ lists the decisions made by the hash function as it scans $x$ using the underlying function $\rho$.  The \defn{grid walk} of $x$, $y$, and $\rho$ is a sequence based on the transcripts (under $\rho$)  of $x$ and $y$---it consists of some edits, and some extra operations that help keep track of how the hashes of $x$ and $y$ interact.  Finally, the \defn{transformation} of $x$, $y$, and $\rho$ is a sequence of edits based on the grid walk of $x$, $y$, and $\rho$.

Using these three sequences, we can set up the basic structure to bound the probability that $x$ and $y$ hash together using their edit distance.  We use these definitions to analyze the probability of collision in \secref{collisionprobs}.

\paragraph{Transcripts.}
A \defn{transcript} is a sequence of hash operations: each element of the sequence is a hash-insert, hash-replace, or hash-match.  Essentially, the transcript of $x$ and $\rho$, denoted $\tau(x,\rho)$, is a log of the actions taken by the hash on string $x$ using underlying function $\rho$.

We define an index function $i(x,k,\rho)$.  The idea is that $i(x,k,\rho)$ is the value of $i$ when the $k$th hash character is written, if hashing $x$ using underlying function $\rho$.

We set $i(x,0,\rho) = 0$ for all $x$ and $\rho$.
Let $(r_{1,k},r_{2,k}) = \rho(x_{i(x,k,\rho)},k)$.  
We can now recursively define both $\tau(x,\rho)$ and $i(x,k,\rho)$.  We denote the $k$th character of $\tau(x,\rho)$ using $\tau_k(x,\rho)$. 
\begin{itemize}[topsep=0pt,noitemsep]
  \item If $r_{1,k} \leq p_a$, then $i(x,k+1,\rho) = i(x,k,\rho)$, and $\tau_k(x,\rho) = \text{hash-insert}$.  
  \item If $r_{1,k} > p_a$ and $r_{2,k}\leq p_r$, then $i(x,k+1,\rho) = i(x,k,\rho) + 1$, and $\tau_k(x,\rho) = \text{hash-replace}$.  
  \item If $r_{1,k} > p_a$ and $r_{2,k} >  p_r$, then $i(x,k+1,\rho) = i(x,k,\rho) + 1$, and $\tau_k(x,\rho) = \text{hash-match}$.  
\end{itemize} 

A transcript $\tau(x,\rho)$ is \defn{complete} if $|\tau(x,\rho)| < 8d/(1-p_a) + 6\log n$.  

%


\begin{lemma}
  \lemlabel{shorthash}
  For any string $x$ of length $d$,
  $ \Pr_\rho[\tau(x,\rho)\text{ is complete}] \geq 1 - 1/n^2 $.
\end{lemma}

\begin{proof}
  If $\tau(x,\rho)$ has $\ell$ hash-insert operations, then $|\tau(x,\rho)| \leq d + \ell$.  We bound the probability that $\ell > 7d/(1 - p_a) + 6\log n$.

  For each character in $x$, we can model the building of $\tau(x,\rho)$ as a series of independent coin flips. On heads (with probability $p_a$), $\ell$ increases; on tails the process stops.  Thus we expect $1/(1-p_a)$ hash-insert operations for each character of $x$, and at most $d/(1 - p_a)$ hash-insert operations overall.

  Using standard Chernoff bounds (i.e.~\cite[Exercise 4.7]{MitzenmacherUpfal17}), the probability that $\ell > 7d/(1 - p_a) + 6\log n$ is at most $\exp((6d(1-p_a) + 6\log n)/3) < 1/n^2$.
\end{proof} 

\paragraph{Grid Walks.}  A \defn{grid walk} $g(x,y,\rho)$ for two strings $x$ and $y$ and underlying function $\rho$ is a sequence that helps us examine how $h_{\rho}(x)$ and $h_{\rho}(y)$ interact---it is a bridge between the transcript of $x$, $y$ and $\rho$, and the transformation induced by $x$, $y$, and $\rho$ (which is a sequence of edits).
We formally define the grid walk, and discuss how it corresponds to a random walk in a graph.  This graph is closely based on the dynamic programming table for $x$ and $y$.

The grid walk is a sequence of length $\max\{|\tau(x,\rho)|,|\tau(y,\rho)|\}$.  The grid walk has an alphabet of size 6: each character is one of $\{\ins, \del,\rep,\loo,\match,\die\}$.  
At a high level, 
\ins, \del, and \rep{} correspond to string edits---for example, \ins{} corresponds to the index of $x$ being incremented while the index of $y$ stays the same (as if we inserted the corresponding character into $y$).  \loo{} corresponds to both strings writing $\bot$ without increasing $i$; the process ``loops'' and we continue with nothing changed except the length of the hash.  \match{} corresponds to the case when both hashes simultaneously evaluate the same character---after a sequence of \loo{} operations, they will \match{} by both writing out either the matching character or $\bot$ to their respective hashes.  \die{} is a catch-all for all other cases: the strings write out different characters, the hashes are no longer equal, and the analysis stops.

We define a directed graph $G(x,y)$ to help explain how to construct the walk.  Graph $G(x,y)$ is a directed graph with $|x||y|+ 1$ nodes, corresponding roughly to the dynamic programming table between $x$ and $y$. We label one node as the \defn{\die{} node}.  We label the other $|x||y|$ nodes using two-dimensional coordinates $(i,j)$ with $0\leq i < |x|$, and $0\leq j< |y|$.

\begin{figure}
  \centering
\begin{tikzpicture}
  \def\gridsize{1.5}
  \begin{scope}[every node/.style={circle,thick,draw}] 
    \node[label={[label distance=-5pt]$(i,j)$}] (A) at (0, \gridsize) {};
    \node[label={[label distance=-5pt]330:$(i+1,j)$}] (B) at (\gridsize, \gridsize) {};
    \node[label={[label distance=-20pt]330:$(i+1,j+1)$}] (C) at (\gridsize, 0) {};
    \node[label={[label distance=-15pt]330:$(i,j+1)$}] (D) at (0, 0) {};
    \node[label={[label distance=-5pt]330:\die}] (E) at (-\gridsize/4,-\gridsize/2.3) {};
  \end{scope} 
    \path [->,draw,thick,sloped,anchor=south,auto=false] 
      (A) edge node {\small \del} (B)
      (A) edge node {\small \rep} (C)
      (A) edge node {\small \ins} (D)
      (A) edge [bend right] node {\small \die} (E); 
    \path [->,draw,thick,auto] 
      (A) edge [loop left] node {\small \loo} (A);
      \node[label={(a)}] at (-\gridsize,-\gridsize/1.2) {};
\end{tikzpicture}
\begin{tikzpicture}
  \def\gridsize{1.5}
  \begin{scope}[every node/.style={circle,thick,draw}] 
    \node[label={[label distance=-5pt]$(i,j)$}] (A) at (0, \gridsize) {};
    \node[label={[label distance=-5pt]330:$(i+1,j)$}] (B) at (\gridsize, \gridsize) {};
    \node[label={[label distance=-20pt]330:$(i+1,j+1)$}] (C) at (\gridsize, 0) {};
    \node[label={[label distance=-15pt]330:$(i,j+1)$}] (D) at (0, 0) {};
    \node[label={[label distance=-5pt]330:\die}] (E) at (-\gridsize/4,-\gridsize/2.3) {};
  \end{scope} 
    \path [->,draw,thick,sloped,anchor=south,auto=false] 
      (A) edge node {\small \match} (C);
    \path [->,draw,thick,auto] 
      (A) edge [loop left] node {\small \loo} (A);
      \node[label={(b)}] at (-\gridsize,-\gridsize/1.2) {};
\end{tikzpicture}
  \vspace{-.5cm}
    \caption{The edges for a single node $(i,j)$ with $i < |x|-1$ and $j < |y| - 1$. (a) represents the edges if $x_i\neq y_j$; (b) represents the edges if $x_i = y_j$.}
  \figlabel{singlenodeedges}
    \vspace{-.7cm}
\end{figure}
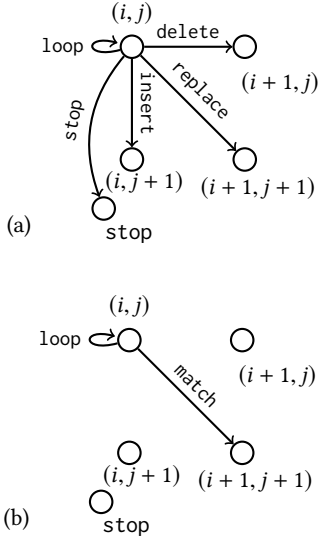

We now list all arcs between nodes.   We label each with a grid walk character; this will be useful for analyzing $g(x,y,\rho)$. 
Consider an $(i,j)$ with $0\leq i< |x|-1$ and $0\leq j< |y|-1$.  For any $(i,j)$ with $x_i\neq y_j$, we place five arcs:
\begin{itemize}[topsep=0pt,noitemsep]
  \item a \del{} arc from $(i,j)$ to $(i+1,j)$, 
  \item a \rep{} arc from $(i,j)$ to $(i+1,j+1)$,
  \item an \ins{} arc from $(i,j)$ to $(i,j+1)$, 
  \item a \loo{} arc from $(i,j)$ to $(i,j)$, and
  \item a \die{} arc from $(i,j)$ to the \die{} node.
\end{itemize}
These arcs are shown in \figref{singlenodeedges}a.
For any $(i,j)$ with $x_i = y_j$, we place two edges: a \match{} arc from $(i,j)$ to $(i+1,j+1)$, and a \loo{} arc from $(i,j)$ to $(i,j)$%
; see \figref{singlenodeedges}b.

  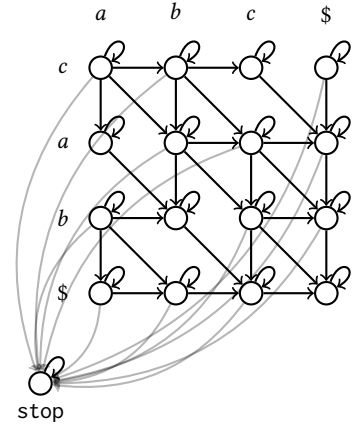
\begin{figure}
    \centering
\begin{tikzpicture}
    \foreach \i in {0,...,3}
       \foreach \j in {0,...,3}
          \node[circle,thick,draw] (\i\j) at (\i,3-\j) {};

    \node[label=270:\die,circle,thick,draw] (stop) at (-.8,-1.2) {};

    \node (x1) at (0,3.7) {$a$};
    \node (x2) at (1,3.75) {$b$};
    \node (x3) at (2,3.7) {$c$};
    \node (x4) at (3,3.7) {$\$$};

    \node (y1) at (-.5,3) {$c$};
    \node (y2) at (-.5,2) {$a$};
    \node (y3) at (-.5,1) {$b$};
    \node (y4) at (-.5,0) {$\$$};

    \foreach \i in {0,...,2}
       \foreach \j in {0,...,2}
          \path[->,draw,thick,auto] (\i\j) edge node {} (\the\numexpr\i+1\relax\the\numexpr\j+1\relax);

    \foreach \i in {0,...,3}
       \foreach \j in {0,...,3}
          \path[->,draw,thick,distance=4mm,in=30,out=70] (\i\j) edge node {} (\i\j);

    \path[->,draw,thick,distance=4mm,in=30,out=70] (stop) edge node {} (stop);

       \path[->,draw,thick] (00) edge node {} (01);
       \path[->,draw,thick] (10) edge node {} (11);
       \path[->,draw,thick] (11) edge node {} (12);
       \path[->,draw,thick] (21) edge node {} (22);
       \path[->,draw,thick] (02) edge node {} (03);
       \path[->,draw,thick] (22) edge node {} (23);
       \path[->,draw,thick] (30) edge node {} (31);
       \path[->,draw,thick] (31) edge node {} (32);
       \path[->,draw,thick] (32) edge node {} (33);

       \path[->,draw,thick] (00) edge node {} (10);
       \path[->,draw,thick] (10) edge node {} (20);
       \path[->,draw,thick] (11) edge node {} (21);
       \path[->,draw,thick] (21) edge node {} (31);
       \path[->,draw,thick] (02) edge node {} (12);
       \path[->,draw,thick] (22) edge node {} (32);
       \path[->,draw,thick] (03) edge node {} (13);
       \path[->,draw,thick] (13) edge node {} (23);
       \path[->,draw,thick] (23) edge node {} (33);

       \path[->,draw,thick,opacity=.29] (00) edge [bend right] node {} (stop);
       \path[->,draw,thick,opacity=.29] (10) edge [bend right] node {} (stop);
       \path[->,draw,thick,opacity=.29] (11) edge [bend right] node {} (stop);
       \path[->,draw,thick,opacity=.29] (21) edge [bend right] node {} (stop);
       \path[->,draw,thick,opacity=.29] (02) edge [bend right] node {} (stop);
       \path[->,draw,thick,opacity=.29] (22) edge [bend left] node {} (stop);
       \path[->,draw,thick,opacity=.29] (30) edge [bend left] node {} (stop);
       \path[->,draw,thick,opacity=.29] (31) edge [bend left] node {} (stop);
       \path[->,draw,thick,opacity=.29] (32) edge [bend left] node {} (stop);
       \path[->,draw,thick,opacity=.29] (03) edge [bend left] node {} (stop);
       \path[->,draw,thick,opacity=.29] (13) edge [bend left] node {} (stop);
       \path[->,draw,thick,opacity=.29] (23) edge [bend left] node {} (stop);

\end{tikzpicture}
    \caption{This figure shows $G(x,y)$ for $x = abc\$ $ and $y = cab\$ $.  For clarity, all edge labels are ommited and \die{} edges are partially transparent.}
    \figlabel{graphexample}
\end{figure}

The rightmost and bottommost nodes of the grid are largely defined likewise, but arcs that lead to nonexistant nodes instead lead to the \die{} node.\footnote{Since $x$ and $y$ are $\$$-terminal, these nodes never satisfy $x_i = y_j$ except at $(|x|-1,|y|-1)$}  
For $0\leq j < |y|-1$
there is an \ins{} arc from $(|x|-1,j)$ to $(|x|-1,j+1)$ a \die{} arc, \del{} arc, and \rep{} arc from $(|x|-1,j)$ to the \die{} node, and a \loo{} arc from $(|x|-1,j)$ to $(|x|-1,j)$.  
For $0\leq i< |x|-1$, there is 
a \del{} arc from $(i,|y|-1)$ to $(i+1,|y|-1)$, a \die{} arc, an \ins{} arc, and a \rep{} arc from $(i,|y|-1)$ to the \die{} node, and a \loo{} arc from $(i,|y|-1)$ to $(i,|y|-1)$.  Finally, node $(|x|-1,|y|-1)$ has a \loo{} arc to $(|x|-1,|y|-1)$.
See \figref{graphexample}.

The stop node has (for completeness) six self loops with labels \match{}, \ins{}, \rep{}, \del{}, \loo{}, and \die{}.

\begin{table}
  \centering
\begin{tabular}{|c|c|c|}
  \hline
    $\tau_k(x,\rho)$ & $\tau_k(y,\rho)$ & $g_k(x,y,\rho)$ \\
    \hline 
    \hline
    hash-replace & hash-replace & \rep \\
    \hline
    hash-replace & hash-insert & \del \\
    \hline
    hash-insert & hash-replace & \ins \\
    \hline
    hash-insert & hash-insert & \loo \\
    \hline
    hash-match & - & \die \\
    \hline
     - & hash-match & \die \\
    \hline 
\end{tabular}
  \caption{This table defines a grid walk for non-matching characters in strings $x$ and $y$, given the corresponding transcripts.}
\tablabel{gridwalk}
  \vspace{-.8cm}
\end{table}

  We now define the grid walk $g(x,y,\rho)$.  We will use $G(x,y)$ to relate $g(x,y,\rho)$ to $h_{\rho}(x)$ and $h_{\rho}(y)$ in Lemmas~\ref{lem:walkcorrectness} and~\ref{lem:walkcollision}.

  We determine the $k$th character of $g(x,y,\rho)$, denoted $g_k(x,y,\rho)$, using $\tau_k(x,\rho)$ and $\tau_k(y,\rho)$, as well as $x_{i(x,k,\rho)}$ and $y_{i(y,k,\rho)}$.  For $k > \min\{|\tau(x,\rho)|,|\tau(y,\rho)|\}$, $g_k(x,y,\rho) = \die$.

  If $x_{i(x,k,\rho)} \neq y_{i(y,k,\rho)}$, we define $g_k(x,y,\rho)$ using \tabref{gridwalk}.  
  
  If $x_{i(x,k,\rho)} = y_{i(y,k,\rho)}$, then $\tau_k(x,\rho) = \tau_k(y,\rho)$.  If $\tau_k(x,\rho)= \tau_k(y,\rho)$ is a hash-insert, then $g_k(x,y,\rho) = \loo$; otherwise, $g_k(x,y,\rho) = \match{}$.

  We say that a grid walk is \defn{complete} if 
  both $\tau(x,\rho)$ and $\tau(y,\rho)$ are complete.
  We say that a grid walk is \defn{alive} if it is complete and it does not contain \die{}.

  The next lemma motivates this definition: the grid walk defines a path through the grid corresponding to the hashes of $x$ and $y$.
  \begin{lemma}
    \lemlabel{walkcorrectness}
    Consider a walk through $G(x,y)$ which at step $i$ takes the edge with label corresponding to $g_i(x,y,\rho)$. 
    Assume $k$ is such that the prefix $g(x,y,\rho)[k]$ of length $k$ is alive.   
    Then after $k$ steps, the walk arrives at node $(i(x,k,\rho),i(y,k,\rho))$.  
  \end{lemma}

\begin{proof}
Our proof is by induction on $k$.  We prove both that the walk arrives at node $(i(x,k,\rho),i(y,k,\rho))$, and that the walk is well-defined: the next character in $g(x,y,\rho)$ always corresponds to an outgoing edge of the current node.   

    For the base case $k = 0$ the proof is immediate, since $(i(x,0,\rho),i(y,0,\rho)) = (0,0)$.  Furthermore, node $(0,0)$ has an outgoing \match{} edge if and only if $x_0 = y_0$ (otherwise it has an outgoing \ins{}, \del{}, and \rep{} edge); similarity, $g_0(x,y,\rho) = \match{}$ only if $x_0 = y_0$ (the rest of the cases follow likewise). 

    Assume that after $k-1$ steps, the walk using $g(x,y,\rho)[k-1]$ arrives at node $(i(x,k-1,\rho),i(y,k-1,\rho))$.  We begin by proving that the walk remains well-defined.  We have $g_{k-1}(x,y,\rho) = \match{}$ only if $x_{i(x,k-1,\rho)} = y_{i(y,k-1,\rho)}$; in this case $(i(x,k-1,\rho),i(y,k-1,\rho))$ has an outgoing \match{} edge.  We have $g_{k-1}(x,y,\rho) = \ins{}$ (or \del{} or \rep{}) only if $x_{i(x,k-1,\rho)} \neq y_{i(y,k-1,\rho)}$; again, node $(i(x,k-1,\rho),i(y,k-1,\rho))$ has the corresponding outgoing edge.  All nodes have outgoing \loo{} and \die{} edges.

    Now we show that after $k$ steps, the walk using $g(x,y,\rho)[k]$ arrives at node $(i(x,k,\rho),i(y,k,\rho))$.  We split into five cases based on $g_{k-1}(x,y,\rho)$ (if $g_{k-1}(x,y,\rho) = \die{}$ the lemma no longer holds).
    \begin{itemize}
      \item[\rep{}] We have $\tau_k(x,\rho) = $ hash-replace, and $\tau_k(y,\rho) = $ hash-replace.  Thus, $i(x,k,\rho) = i(x,k-1,\rho) + 1$ and $i(y,k,\rho) = i(y,k-1,\rho) + 1$.  In $G(x,y)$, the edge labelled \rep{} leads to node $(i(x,k-1,\rho)+1,i(y,k-1,\rho) + 1)$.
      \item[\match{}] We have $\tau_k(x,\rho) = $ hash-replace, and $\tau_k(y,\rho) = $ hash-replace.  Thus, $i(x,k,\rho) = i(x,k-1,\rho) + 1$ and $i(y,k,\rho) = i(y,k-1,\rho) + 1$.  In $G(x,y)$, the edge labelled \match{} leads to node $(i(x,k-1,\rho)+1,i(y,k-1,\rho) + 1)$.
      \item[\del{}] We have $\tau_k(x,\rho) = $ hash-replace, and $\tau_k(y,\rho) = $ hash-insert.  Thus, $i(x,k,\rho) = i(x,k-1,\rho) + 1$ and $i(y,k,\rho) = i(y,k-1,\rho)$.  In $G(x,y)$, the edge labelled \ins{} leads to node $(i(x,k-1,\rho)+1,i(y,k-1,\rho))$.
      \item[\ins{}] We have $\tau_k(x,\rho) = $ hash-insert, and $\tau_k(y,\rho) = $ hash-replace.  Thus, $i(x,k,\rho) = i(x,k-1,\rho)$ and $i(y,k,\rho) = i(y,k-1,\rho)+1$.  In $G(x,y)$, the edge labelled \ins{} leads to node $(i(x,k-1,\rho),i(y,k-1,\rho)+1)$.
      \item[\loo{}] We have $\tau_k(x,\rho) = $ hash-insert, and $\tau_k(y,\rho) = $ hash-insert.  Thus, $i(x,k,\rho) = i(x,k-1,\rho)$ and $i(y,k,\rho) = i(y,k-1,\rho)$.  In $G(x,y)$, the edge labelled \loo{} leads to node $(i(x,k-1,\rho),i(y,k-1,\rho))$.\qedhere
    \end{itemize}
\end{proof}

  With this in mind, we can relate grid walks to hash collisions.

  \begin{lemma}
    \lemlabel{walkcollision}
    Let $x$ and $y$ be any two strings, and $\rho$ be any underyling function where both $\tau(x,\rho)$ and $\tau(y,\rho)$ are complete.

    Then $h_\rho(x) = h_{\rho}(y)$ if and only if $g(x,y,\rho)$ is alive.  Furthermore, if $h_\rho(x) = h_\rho(y)$ then the path defined by $g(x,y,\rho)$ reaches node $(|x|,|y|)$.
  \end{lemma}

  \begin{proof}
    \ifdirection Assume that $h_{\rho}(x) = h_{\rho}(y)$; we show that the path defined by $g(x,y,\rho)$ is alive and reaches $(|x|,|y|)$. 
    
    First, $g(x,y,\rho)$ must be alive: $g_k(x,y,\rho) = \die{}$ only when $x_{i(x,k,\rho)} \neq y_{i(y,k,\rho)}$ and either $\tau_k(x,\rho) = $ hash-match or $\tau_k(y,\rho) = $ hash-match, or when $k > \min\{|\tau(x,\rho)|,|\tau(y,\rho)|\}$. Since $x_{i(x,k,\rho)}$ (resp. $y_{i(y,k,\rho)}$) is appended to the hash on a hash-match, this contradicts $h_{\rho}(x) = h_{\rho}(y)$.  Furthermore, we must have $|\tau(x,\rho)| = |\tau(y,\rho)|$ because 
    $|\tau(x,\rho)| = |h_{\rho}(x)| = |h_{\rho}(y)| = |\tau(y,\rho)|$.  

    Since $\tau(x,\rho)$ and $\tau(y,\rho)$ are complete, 
    $i(x,|\tau(x,\rho)|-1,\rho) = |x|$ and $i(y,\tau(y,\rho)-1,\rho) = |y|$.  Thus, by \lemref{walkcorrectness}, the walk reaches $(|x|,|y|)$.  

    \onlyifdirection We show that if $h_{\rho}(x) \neq h_{\rho}(y)$ then $g(x,y,\rho)$ is not alive.  Let $k$ be the smallest index such that the $k$th character of $h_{\rho}(x)$ is not equal to the $k$th character of $h_{\rho}(y)$.  At least one of these characters cannot be $\bot$; thus either $\tau_k(x,\rho) = $ hash-match, or $\tau_k(y,\rho) = $ hash-match.  If $x_{i(x,k,\rho)} \neq y_{i(x,k,\rho)}$, then $g_k(x,y,\rho) = \die{}$ and we are done.  Otherwise, $x_{i(x,k,\rho)} = y_{i(y,k,\rho)}$; thus $\tau_k(x,\rho) = \tau_k(y,\rho)$, and the $k$th character of both $h_{\rho}(x)$ and $h_{\rho}(y)$ is $x_{i(x,k,\rho)} = y_{i(y,k,\rho)}$.  But this contradicts the definition of $k$. 
  \end{proof}


  We now bound the probability that the grid walk traverses each edge in $G(x,y)$.  

\begin{lemma}
  \lemlabel{gridwalkprobs}
  Let $x$ and $y$ be any two strings, and for any $k < 8d/(1-p_a) + 6\log n$  let $E_k$ be the event that $i(x,k,\rho) < |x|$, $i(y,k,\rho) < |y|$, and $x_{i(x,k,\rho)} \neq y_{i(y,k,\rho)}$. 
Then if $\Pr_\rho[E_k] > 0$, the following four conditional bounds hold:
      \begin{align*} 
        \Pr_\rho[g_k(x,y,\rho) = \loo{} ~|~ E_k] &= p_a^2 \\
        \Pr_\rho[g_k(x,y,\rho) = \del{} ~|~ E_k] &= p_a(1-p_a)p_r  \\
        \Pr_\rho[g_k(x,y,\rho) = \ins{} ~|~ E_k] &= p_a(1-p_a)p_r \\
        \Pr_\rho[g_k(x,y,\rho) = \rep{} ~|~ E_k] &= (1-p_a)^2p_r^2.
      \end{align*} 
\end{lemma}

  \begin{proof}
  We have $|\tau(x,\rho)| > k$ and $|\tau(y,\rho)| > k$ from $E_k$.  
  Thus: 
  \begin{itemize}[topsep=0pt,noitemsep]
    \item $\Pr_\rho(\tau_k(x,\rho) = \text{hash-insert} ~|~ E_k) = p_a$
    \item $\Pr_\rho(\tau_k(x,\rho) = \text{hash-replace} ~|~ E_k) = (1-p_a)p_r$
    \item $\Pr_\rho(\tau_k(x,\rho) = \text{hash-match} ~|~ E_k) = (1-p_a)(1-p_r)$.
  \end{itemize}
  The respective probabilities for $\tau_k(y,\rho)$ hold as well.
  Combining these probabilities with \tabref{gridwalk} gives the lemma.  
  \end{proof}

\paragraph{Transformations.}  
We call a sequence of edits for a pair of strings $x$ and $y$ \defn{greedy} if they can be applied to $x$ in order from left to right,  and all operations are performed on non-matching positions.   We formally define this in \defref{transformation}. 
 With this in mind, we can simplify a sequence of edits for a given $x$ and $y$, with the understanding that they will be applied greedily. 

A \defn{transformation} is a sequence of edits with position and character information removed: it is a sequence consisting only of \texttt{insert}, \texttt{delete}, and \texttt{replace}.  We let $T(x,y)$ be the string that results from greedily applying the edits in $T$ to $x$ when $x$ does not match $y$. 
We say that a transformation is \defn{valid} for strings $x$ and $y$ if the total number of \del{} or \rep{} operations in $T$ is at most $|x|$, and the total number of \ins{} or \rep{} operations in $T$ is at most $|y|$. 
The following definition formally defines how to apply these edits.

 \begin{definition}
   \deflabel{transformation}
   Let $x$ and $y$ be two $\$$-terminal strings, and let $T$ be a transformation that is valid for $x$ and $y$.

   If $T$ is empty, $T(x,y) = x$.  Otherwise we define $T(x,y)$ inductively.  Let $T' = T[|T| - 1]$ be $T$ with the last operation removed, let $\sigma = T_{|T| - 1}$ be the last operation in $T$, and let $i$ be the smallest index such that the $i$th character of $T'(x,y)$ is not equal to $y_i$.  Position $i$ always exists if $T'(x,y) \neq y$ because $x$ and $y$ are $\$ $-terminal; otherwise $i = 0$.\footnote{The case where $i$ is reset to $0$ is included for completeness and will not be used in the rest of the paper.  It only occurs when $x$ is first transformed into $y$, and then a sequence of redundant edits (such as an equal number of inserts and deletes) are performed.}

   We split into three cases depending on $\sigma$.  If $\sigma = \ins$, we obtain $T(x,y)$ by inserting $y_i$ at position $i$ in $T'(x,y)$.  If $\sigma = \del$, we obtain $T(x,y)$ by deleting the $i$th character of $T'(x,y)$.  Finally, if $\sigma = \rep$, we obtain $T(x,y)$ by replacing the $i$th character of $T'(x,y)$ with $y_i$.  
 \end{definition}

 We say that a transformation $T$ \defn{solves} $x$ and $y$ if $T$ is valid for $x$ and $y$, $T(x,y) = y$, and for any $i < |T|$, the prefix $T' = T[i]$ satisfies $T'(x,y)\neq y$. 

 A classic observation is that edit distance operations can be applied from left to right, greedily skipping all matches.  The following lemma shows that this intuition applies to transformations.  Since \defref{transformation} does not allow characters to be appended onto the end of $x$, we use the appended character $\$ $ to ensure that there is an optimal transformation between any pair of strings.  

 \begin{lemma}
   \lemlabel{transformationfromED}
   Let ${x}$ and ${y}$ be two strings that do not contain $\$ $.  Then if $\ED(x,y) = r$,
   \begin{itemize}[topsep=0pt,noitemsep]
     \item there exists a transformation $T$ of length $r$ that solves $x\concat \$$ and $y\concat\$$, and
     \item there does not exist any transformation $T'$ of length $<r$ that solves $x\concat \$$ and $y\concat\$$. 
   \end{itemize}
 \end{lemma}

 \begin{proof}
   We prove a single statement implying the lemma: if $\hat{T}$ is the shortest transformation that solves $x\concat\$ $ and $y\concat\$ $, then $|\hat{T}| = \ED(x,y)$.

   Let $\sigma_1,\ldots \sigma_r$ be the sequence of edits applied to $x\concat\$ $ to obtain $\hat{T}(x\concat\$,y\concat\$)$ in \defref{transformation}.  These operations apply to increasing indices $i$ because $\hat{T}$ is the shortest transformation satisfying $\hat{T}(x\concat\$,y\concat\$)$.  Let $\sigma_i$ be the last operation that applies to an index $i < |x|$. Let $\hat{y} = \hat{T}[i+1](x\concat\$,y\concat\$)$ be the string obtained after applying the operations of $\hat{T}$ through $\sigma_i$.  Clearly, $x$ is a prefix of $\hat{y}$.  
   We claim that because $\hat{T}$ is the shortest transformation, the operations in $\hat{T}$ after $\sigma_i$ must be $|\hat{y}| - |x| - 1|$ \ins{} operations.  Clearly there must be at least $|\hat{y} - |x| - 1|$ operations after $\sigma_i$ because $i$ is increasing and only one character in $\hat{y}$ matches the final character $\$$ of $x$.  By the same argument, if $\hat{T}$ has any \ins{} or \rep{} operations it cannot meet this bound.  

   With this we have $\ED(x,y) \leq |\hat{T}|$ because we can apply $\sigma_1,\ldots,\sigma_i$, followed by $|\hat{y} - |x| - 1|$ insert operations to $x$ to obtain $y$.  This totals to $|\hat{T}|$ operations overall.

   We also have $\ED(x,y) \geq |\hat{T}|$ by minimality of $\hat{T}$ because any sequence of edits applied to $x$ that obtains $y$ will obtain $y\concat \$$ when applied to $x\concat\$$.  
\end{proof}

For a given $x$, $y$, and $\rho$, we obtain the \defn{transformation induced by $x$, $y$, and $\rho$}, denoted $\calT (x,y,\rho)$, by removing all occurrences of \loo{} and \match{} from $g(x,y,\rho)$ if $g(x,y,\rho)$ is alive.   Otherwise, $\calT(x,y,\rho)$ is the empty string.  

In \lemref{hasheseqtransformation} we show that strings $x$ and $y$ collide exactly when their induced transformation $T$ solves $x$ and $y$.  This can be seen intuitively in \figref{graphexample}---the grid walk is essentially a random walk through the dynamic programming table.


\begin{lemma}
  \lemlabel{hasheseqtransformation}
Let $x$ and $y$ be two distinct strings and let $T = \calT(x,y,\rho)$.  
  Then $h_{\rho}(x) = h_{\rho}(y)$ if and only if $T$ solves $x$ and $y$.
\end{lemma}

\begin{proof} 
  \ifdirection Assume $T$ solves $x$ and $y$.  Since $x\neq y$, $T$ must be nonempty; thus $g(x,y,\rho)$ is alive.  By \lemref{walkcollision}, $h_\rho(x) = h_\rho(y)$.  

    \onlyifdirection Assume $h_\rho(x) = h_\rho(y)$; by \lemref{walkcollision} $g(x,y,\rho)$ is alive.

    Let $g(x,y,\rho)[k]$ be the prefix of $g(x,y,\rho)$ of length $k$, and let $T^k$ be $g(x,y,\rho)[k]$ with $\loo{}$ and $\match{}$ removed.  We prove by induction that $T^k(x[i(x,k,\rho)],y[i(y,k,\rho)]) = y[i(y,k,\rho)]$.  This is trivially satisfied for $k = 0$.  

    Assume that $T^{k-1}(x[i(x,k-1,\rho)],y[i(y,k-1,\rho)]) = y[i(y,k-1,\rho)]$.  We split into five cases based on the $k$th operation in $g(x,y,\rho)$.
\begin{itemize}
  \item[\match] We must have $x_{i(x,k-1,\rho)} = y_{i(y,k-1,\rho)}$ and $T^k = T^{k-1}$.  Furthermore, $i(x,k,\rho) = i(x,k-1,\rho) + 1$ and $i(y,k,\rho) = i(y,k-1,\rho) + 1$.  Thus $T^{k}(x[i(x,k,\rho)],y[i(y,k,\rho)]) = T^{k-1}(x[i(x,k-1,\rho)],y[i(y,k-1,\rho)])\concat x_{i(x,k-1,\rho)} = y[i(y,k,\rho)]$.  
  \item[\ins] We have $i(x,k,\rho) = i(x,k-1,\rho)$ and $i(y,k,\rho) = i(y,k-1,\rho) + 1$.  Thus, $T^{k-1}(x[i(x,k,\rho)],y[i(y,k,\rho)]) = y[i(y,k-1,\rho)]$ and $y[i(y,k,\rho)]$ differ only in the last character.  Then 
    $T^{k}(x[i(x,k,\rho)],y[i(y,k,\rho)]) = T^{k-1}(x[i(x,k-1,\rho)],y[i(y,k-1,\rho)])\concat y_{i(y,k,\rho)-1} = y[i(y,k,\rho)]$.  

  \item[\rep] We have $i(x,k,\rho) = i(x,k-1,\rho)+1$ and $i(y,k,\rho) = i(y,k-1,\rho) + 1$.  Thus, $T^{k-1}(x[i(x,k,\rho)],y[i(y,k,\rho)]) = y[i(y,k-1,\rho)] \concat x_{i(x,k,\rho)-1}$ and $y[i(y,k,\rho)]$ differ only in the last character.  By definition, the final character of $T^{k-1}(x[i(x,k,\rho)],y[i(y,k,\rho)]$ is replaced with $y_{i(y,k,\rho)-1}$, obtaining $T^k(x[i(x,k,\rho)],y[i(y,k,\rho)]) = y[i(y,k-1,\rho)]\concat y_{i(y,k,\rho)-1} = y[i(y,k,\rho)]$.  

  \item[\del] We have $i(x,k,\rho) = i(x,k-1,\rho)+1$ and $i(y,k,\rho) = i(y,k-1,\rho)$.  Thus, $T^{k-1}(x[i(x,k,\rho)],y[i(y,k,\rho)]) = y[i(y,k-1,\rho)] \concat x_{i(x,k,\rho) - 1}$ and $y[i(y,k,\rho)]$ differ only in the last character (which is deleted). Then 
    $T^{k}(x[i(x,k,\rho)],y[i(y,k,\rho)]) = T^{k-1}(x[i(x,k-1,\rho)],y[i(y,k-1,\rho)]) = y[i(y,k,\rho)]$.  

  \item[\loo] We have $i(x,k,\rho) = i(x,k-1,\rho)$, $i(y,k,\rho) = i(y,k-1,\rho)$, and $T^k = T^{k-1}$.  We immediately obtain $T^k(x[i(x,k,\rho)],y[i(y,k,\rho)]) = y[i(y,k,\rho)]$.  
\end{itemize}
  By \lemref{walkcollision}, $g(x,y,\rho)$ reaches node $(|x|,|y|)$, so the above shows that with $k = |g(x,y,\rho)|$, $T(x,y) = y$.  
\end{proof}

We are finally ready to prove \lemref{probinduces}, which forms the basis of our performance analysis.


\begin{lemma}
  \lemlabel{probinduces}

  For any $\$$-terminal strings $x$ and $y$, let $T$ be a transformation of length $t$ that is valid for $x$ and $y$.  
  Then
  \[
    p^t - 1/n^2 \leq \Pr_\rho[T \text{ is a prefix of } \calT(x,y,\rho)] \leq p^t.
  \]
\end{lemma}

\begin{proof}
  Let $G_T$ be the set of all grid walks $g$ (i.e. the set of all sequences consisting of grid walk operations) such that $g$ does not contain \die{}, and deleting \loo{} and \match{} from $g$ results in $T_g$ such that $T$ is a prefix of $T_g$.  
  Then by definition, if $T$ is a prefix of $\calT(x,y,\rho)$ then $g(x,y,\rho)\in G_T$; furthermore, if $g(x,y,\rho)\in G_T$ and $g(x,y,\rho)$ is complete, then $T$ is a prefix of $\calT(x,y,\rho)$. 

  We begin by proving that $\Pr_\rho[g(x,y,\rho) \in G_T] = p^t$.  
  We prove this by induction on $t$; $t = 0$ is trivially satisfied.  We assume that $\Pr_\rho[g(x,y,\rho)\in G_{T'}] = p^{t-1}$ for any $G_{T'}$ with $|T'| = t-1$, and prove that it holds for any $T$ with $|T| = t$.

  Let $\sigma$ be the last operation in $T$, and let $T' = T[|T|-1]$ be $T$ with $\sigma$ removed.  
  Thus, $g(x,y,\rho) \in G_T$ only if there exist grid walks $g'$ and $g''$ satisfying 
  \begin{itemize}[topsep=0pt,noitemsep]
    \item $g' \in G_{T'}$,
    \item $g''$ consists of $\loo{}$ and $\match{}$ operations concatenated onto the end of $g'\concat \sigma$, ending with a $\match{}$ operation, and
    \item $g(x,y,\rho)$ consists of zero or more \loo{} operations concatenated onto $g''$.
  \end{itemize}
  By definition of conditional probability, 
  \begin{align*}
    \Pr[g(x,y,\rho) \in G_T] = \qquad \qquad \qquad\qquad \qquad \qquad\qquad \qquad \qquad     \\
    \sum_{g'}\bigg( \Pr[g'\in G_{T'}] \cdot\qquad \qquad \qquad \qquad \qquad \qquad \qquad \qquad \qquad\\
    \sum_{g''} \Pr[g''\in G_T ~|~ g'\in G_{T'} ] \Pr[g(x,y,\rho)\in G_T ~|~ g''\in G_T]\bigg).
  \end{align*}
  We bound these terms one at a time.

  Clearly there is only one $g'$ satisfying the conditions, which can be obtained by taking the prefix of $g(x,y,\rho)$ before the final $\ins{}$, $\del{}$, or $\rep{}$ operation.  By the inductive hypothesis, $\sum_{g'} \Pr[g'\in G_{T'}] = p^{t-1}$.  


  We now bound $\Pr[g''\in G_T ~|~ g'\in G_T]$.  The conditional means that we can invoke \lemref{gridwalkprobs} (as $\Pr[E_k] = p^{t-1} > 0$).  
  
  We have $\Pr[g''\in G_T ~|~ g'\in G_{T'}] = \Pr[g''\in G_T ~|~ g'\concat\sigma\in G_T]\Pr[g'\concat\sigma\in G_T ~|~ g'\in G_{T'}]$.

  We split into two cases depending on $\sigma$.  Recall that $p_r = p_a/(1 - p_a)$.  Since $T$ is valid, if $\sigma = \del{}$ or $\sigma = \rep{}$ we cannot have $i(x,k,\rho) = |x|-1$;  similarly if $\sigma = \ins{}$ or $\sigma = \rep{}$ we cannot have $i(y,k,\rho) = |y| - 1$.  
  Then by \lemref{gridwalkprobs}, if $\sigma = \del{}$ or $\sigma = \ins{}$,
  \[
    \Pr(g'\concat \sigma \in G_T ~|~ g'\in G_{T'}) = p_a(1-p_a)p_r = p_a^2.
  \]
Similarly, if $\sigma = \rep{}$, 
  \[
    \Pr(g'\concat \sigma \in G_T ~|~ g'\in G_{T'}) = (1-p_a)^2p_r^2 = p_a^2.
  \]
  For any $k$ such that $i(x,k,\rho) = i(y,k,\rho)$, $g_k(x,y,\rho)\neq \die{}$ by definition; meanwhile, if $i(x,k,\rho) \neq i(y,k,\rho)$ then $g_k(x,y,\rho)\neq \match{}$.  Thus, $\sum_{g''} \Pr(g''\in G_T ~|~ g'\concat\sigma\in G_T) = 1$.

  Finally we bound $\Pr[g(x,y,\rho)\in G_T ~|~ g''\in G_T]$.  Let $\ell$ be the number of operations concatenated onto $g''$ to obtain $g(x,y,\rho)$.  Then by \lemref{gridwalkprobs}, 
  \[
    \Pr[g(x,y,\rho)\in G_T ~|~ g''\in G_T]  = \sum_{\ell} p_a^{2\ell} = 1/(1 - p_a^2).
  \]

  Multiplying the above bounds, we have $\Pr[g(x,y,\rho) \in G_T] = p^{t-1}p_a^2/(1 - p_a^2)$.  Noting that $p = p_a^2/(1 - p_a^2)$, we obtain $\Pr[g(x,y,\rho)\in G_T] = p^t$.

  We have that if $T$ is a prefix of $\calT(x,y,\rho)$ then $g(x,y,\rho)\in G_T$; thus $\Pr_\rho[T \text{ is a prefix of }\calT(x,y,\rho)] \leq p_t$.

  Meanwhile, $T$ is a prefix of $\calT(x,y,\rho)$ if $g(x,y,\rho)\in G_T$ and $g(x,y,\rho)$ is complete.  
  By the inclusion-exclusion principle,
  \begin{align*}
    \Pr[T \text{ is a prefix of } \calT(x,y,\rho)] =  \qquad\qquad\qquad\qquad\qquad\qquad\\
    \Pr[g(x,y,\rho)\in G_T] + \Pr[g(x,y,\rho) \text{ is complete}] - \\
    \Pr[g(x,y,\rho)\in G_T \text{ or } g(x,y,\rho) \text{ is complete}]\\
    \qquad\qquad\qquad\qquad\qquad\geq p^t + \Pr[g(x,y,\rho) \text{ is complete}] - 1
  \end{align*}
  We have that $\Pr[g(x,y,\rho) \text{ is complete}] = 1 - \Pr[\tau(x,\rho) \text{ or } \tau(y,\rho) \text{ is not complete}]$. By union bound and \lemref{shorthash}, $\Pr[g(x,y,\rho) \text{ is complete}] \geq 1 - 2/n^2$.  Substituting, $\Pr[T\text{ is a prefix of }\calT(x,y,\rho)] \geq p^t - 2/n^2$.
\end{proof}

\subsection{Bounds on Collision Probabilities}
\seclabel{collisionprobs}

We can now bound the probability that two strings collide.  

\begin{lemma}
  \lemlabel{closepoints}
  If $x$ and $y$ satisfy $\ED(x,y) \leq r$, then 
  $ \Pr_{\rho} (h_\rho(x) = h_\rho(y)) \geq p^r - 2/n^2$.
\end{lemma}

\begin{proof}
  Because $\ED(x,y)\leq r$, by \lemref{transformationfromED} there exists a transformation $T$ of length $r$ that solves $x$ and $y$.  By \lemref{probinduces}, $h$ induces $T$ on $x$ and $y$ (which is sufficient for $h(x) = h(y)$ by \lemref{hasheseqtransformation}) with probability $p^r - 2/n^2$.  
\end{proof}


The corresponding upper bound requires that we sum over many possible transformations.  

\begin{lemma}
  \lemlabel{looseupper}
  If $x$ and $y$ satisfy $\ED(x,y) \geq cr$, then 
  $ \Pr_{\rho} (h_\rho(x) = h_\rho(y)) \leq (3p)^{cr} $.
\end{lemma}

\begin{proof}
  Let $\calT$ be the set of all transformations that solve $x$ and $y$.  By 
  \lemref{hasheseqtransformation} and 
  \lemref{probinduces}, 
  \[
    \Pr_{h\in \calH} (h(x) = h(y)) = \sum_{T\in \calT} p^{|T|}.
  \]
Thus, we want to find the $\calT$ (for the given $x$ and $y$) that maximizes this probability.

  
  Since all pairs $T_1, T_2\in \calT$ solve $x$ and $y$, there is no pair $T_1, T_2\in \calT$ such that $T_1$ is a prefix of $T_2$.  Thus, $\calT$ can be viewed as the leaves of a trie of branching factor at most $3$, where each leaf has depth at least $cr$.  
  
  
  We show that without loss of generality all leaves are at depth $cr$.  Consider a leaf $T_1$ at the maximum depth of the trie $i > cr$, and its siblings $T_2$ and $T_3$ if they exist.  Collapse this leaf and its siblings, replacing them instead with a leaf $T_p$ corresponding to their parent in the trie; call the resulting set $\calT'$.  Since we have added a transformation of length $i-1$ and removed at most three of length $i$, this  changes the total cost of $\calT$ by at least $p^{i-1} - 3p^{i}$; this is positive since $p \leq 1/3$.
Repeating this process results in a set $\calT_M$ with all nodes at depth $cr$, where $\calT_M$ gives larger collision probability than the original set $\calT$.  

  There are at most $3^{cr}$ transformations in $\calT_M$, each of length $cr$.  Thus
  \[
    \Pr_{h\in \calH} (h(x) = h(y)) \leq 3^{cr} p^{cr}.\qedhere
  \]
\end{proof}

The following special case is not used in our similarity search bounds, but may be useful in understanding performance on some datasets.

\begin{lemma}
  \lemlabel{nomatchcollisionprob}
  Let $x$ and $y$ be two $\$$-terminal strings with $\ED(x,y) \geq cr$ such that for all $i < |x|-1$ and $j < |y|-1$, $x_i\neq y_j$.  Then $ \Pr_\rho(h_\rho(x) = h_\rho(y)) \leq \left(2p/(1-p)\right)^{cr}$.

\end{lemma}

\begin{proof}
  Let $\hat{x}$ and $\hat{y}$ be arbitrary $\$$-terminal strings of length $cr$ with no other characters in common.  We use grid walks on $G(\hat{x},\hat{y})$ to reason about grid walks on $G(x,y)$.

  Let $GR(i,j)$ be the set of all grid walks reaching node $(i,j)$ in $G(\hat{x},\hat{y})$.  Let $W(i,j) = \Pr[g(\hat{x},\hat{y},\rho)\in GR(i,j)]$.  We have $W(0,0) = 1$.  

  Clearly, $GR(i,j)$ is a subset of $GR(i-1,j) \cup GR(i-1,j-1) \cup GR(i,j-1)$.  In fact, using a case-by-case analysis essentially identical to that of \lemref{probinduces}, 
  \[
    W(i,j) \leq p \cdot W(i-1,j) + p\cdot W(i-1,j-1) + p\cdot W(i,j-1).
  \]
  We take $W(i^*,-1) = 0 = W(-1,j^*)$ for all $i^*$ and $j^*$ so that we can state this recursion without border cases.

  We show by induction that if $\max{i,j} = \ell$, then $W(i,j) \leq (2p/(1-p))^{\ell}$.  This is already satisfied for $\ell = 0$. 

  Assume that the induction is satisfied for all $W(i^*,j^*)$ with $\max\{i^*,j^*\} = \ell - 1$.  For all $(i,j)$ such that $\max\{i,j\} = \ell$, at most two of $(i-1,j-1)$, $(i-1,j)$, and $(i,j-1)$ have max $\ell - 1$; the remaining pair has max $\ell$.  Thus 
  \[
    W(i,j) \leq p\left(\frac{2p}{1-p}\right)^{\ell - 1} + p\left(\frac{2p}{1-p}\right)^{\ell - 1} + p\left(\frac{2p}{1-p}\right)^\ell \leq \left(\frac{2p}{1-p}\right)^{\ell}
  \]

  All grid walks in $G(x,y)$ that go through $(|x|-1,|y|-1)$ must be in $GR(|x|-1,|y|-1)$.  Since we must have $\max\{|x|,|y|\} = cr + 1$, the proof is complete.
\end{proof}

\subsection{Final Running Time for Approximate Similarity Search}
\seclabel{settingnn}

In this section, we describe how to get from our LSH to an algorithm satisfying \defref{ASS}, proving \thmref{similarity}.  

\paragraph{Space and Preprocessing.} 
To preprocess, we first pick $R=\Theta(1/p_1)$ underlying hash functions $\rho_1, \rho_2,\ldots, \rho_{R}$.  For each string $x$ stored in the database, we calculate $h_{\rho_1}(x),\ldots, h_{\rho_{R}}(x)$, and store them in a dictionary data structure for fast lookups (for example, these can be stored in a hash table, where each $h_{\rho}(x)$ has a back pointer to $x$).

We can further decrease the space by storing a random $\log n$-bit hash of $h_{\rho}(x)$ for all $\rho$, rather than the full hash string of length $\Theta(d)$.  
We set $1/p_1 = 3^r n^{1/c}$ (see the discussion below), leading to space $\tilde{O}(3^r n^{1+1/c} + dn)$.  


We store the underlying functions $\rho_1,\rho_2,\ldots \rho_{R}$ so they can be used during queries.  We discuss how this can be achieved without affecting the space bounds in \secref{practical}.

\paragraph{Queries.} 
For a given query $q$, we calculate $h_{1}(q),h_2(q),\ldots,h_{R}(q)$.  For each database string $x$ that collides with $q$ (i.e. for each $x$ such that there exists an $i$ with $h_{\rho_i}(q) = h_{\rho_i}(x)$), we calculate $\ED(x,q)$.  We return $x$ if the distance is at most $cr$.  After repeating this for all $R$ underlying functions, we return that there is no close point.

Correctness of the data structure follows from the definition of $p_1$: if $\ED(q,x) \leq r$, then after $\Theta(1/p_1)$ independent hash functions, $q$ and $x$ collide on at least one hash function with constant probability.


The cost of each repetition is the cost to hash, plus the number of database elements at distance $> cr$ that collide with $q$.  The cost to hash is $O(d/(1-p_a) + \log n)$ by definition, and the cost to test if two strings have distance at most $cr$ is $O(dcr)$ by~\cite{Ukkonen85}.
The number of elements with distance $> cr$ that collide with $q$ is at most $np_2$ in expectation.  Thus our total expected cost can be written
\[
 O\left(\frac{1}{p_1}\left(\frac{d}{1-p_a} + \log n +  (dcr) np_2\right)\right).
\]
This can be minimized (up to a factor $O(\log n)$) by setting $p_2 = 1/ncr$ (recall that $p_a \leq 1/2$).

Thus, we set $p_2 = 1/ncr$, which occurs at $p =  1/(3(ncr)^{1/cr})$. Using this value of $p$, we get $p_1 \geq p^{r} = \Omega(1/(r3^rn^{1/c}))$.  


Putting this all together, the expected query time is $\tilde{O}(d 3^r n^{1/c})$.

\subsection{Approximate Nearest Neighbor}
\seclabel{nearestneighbor}

In this section we generalize \secref{settingnn} to prove \thmref{nearestneighbor}.  
Let $R = \{i\in \{1,\ldots,d\} ~|~ 3^i n^{1/c} \leq n\}$.  We build $O(\log n)$ copies of the data structure described in \secref{settingnn} for each $r^*\in R$.

\paragraph{Queries.} We iterate through each $r^*\in R$ in increasing order, querying the data structure as described above.  If we find a string at distance at most $cr^*$ we stop and return it.  If we reach an $r^*$ such that $3^{r^*}n^{1/c} > n$, we simply scan through all strings to check which is the closest.

Assume the actual nearest neighbor is at distance $r$.  
By Chernoff bounds, we succeed with high probability when $r^*= r$; that is, we return a string at distance at most $cr$.  Thus, the cost is at most $\sum_{r^* = 1}^{r} \tilde{O}(d3^{r^*}n^{1/c}) = \tilde{O}(d3^rn^{1/c})$ with high probability. 

\paragraph{Space.}  We build $O(\log n)$ copies of each data structure; thus the total space is $\sum_{r = 1}^{r^*} \tilde{O}(3^rn^{1 + 1/c} + dn) = \tilde{O}(n^2 + dn)$ by definition of $r^*$. 
We obtain preprocessing time $\tilde{O}(dn^2)$ immediately. 

\subsection{Storing Underlying Functions}
\seclabel{practical}

Our algorithm uses a large number of fully-random, real-number hashes; this causes issues with the space bounds since we need to store each hash.
In this section we relax this assumption.   

We modify $\rho$ to hash to a uniformly random element of the set $\{0, \epsilon, 2\epsilon, \ldots, 1\}$.  
Since the domain of each $\rho$ has size $O(|\Sigma|(d + \log n)$, this means that each $\rho$ can be stored in $O(|\Sigma|\log(1/\epsilon)(d + \log n))$ bits of space.

Intuitively, setting $\epsilon = 1/n$ should not affect our query bounds, while still retaining the space bounds of \thmreftwo{similarity}{nearestneighbor}.  We prove this formally in \lemref{epsilonroundedrho}.

%
%
%


\begin{lemma}
  \lemlabel{epsilonroundedrho}
  With $p_a$ and $p_r$ increased by $\epsilon = 1/n$, and assuming $d = O(n)$, if $x$ and $y$ satisfy $\ED(x,y)\leq r$, $\Pr(h(x) = h(y) \geq \Omega(p^r - 2/n^2)$.  If $x'$ and $y'$ satisfy $\ED(x',y') \geq cr$ then $\Pr((h(x) = h(y)) \leq O((3p)^{cr})$.
\end{lemma}

\begin{proof}
  For simplicity, we let $\hat{p_a} = p_a + \epsilon$ and $\hat{p_r} = p_r + \epsilon$.  

  Since $p_1 = \Omega(1/(r3^r n^{1/c})$, we have (omitting constants for simplicity) $p = 1/(3n^{1/rc})$.  Therefore, $p_a = \sqrt{1/(1 + 3n^{1/rc})} \gg 1/n$, and thus $p_r = p_a/(1- p_a) \gg 1/n$.  Thus, $p_a < \hat{p_a} < p_a(1 + 1/n)$ and $p_r < \hat{p_r} < p_r(1 + 1/n)$.

  Let $\epsilon'$ satisfy
  \begin{equation}
    \eqlabel{perturbedps1}
    p(1 - \epsilon') \leq \frac{\hat{p_a}(1-\hat{p_a})\hat{p_r}}{(1 - \hat{p_a}^2)} \leq p(1 + \epsilon')
  \end{equation}
  and
  \begin{equation}
    \eqlabel{perturbedps2}
    p(1 - \epsilon') \leq \frac{(1-\hat{p_a})^2\hat{p_r}^2}{(1 - \hat{p_a}^2)} \leq p(1 + \epsilon').
  \end{equation}
  Then the proof of \lemref{probinduces} gives that for any $\$$-terminal strings $x$ and $y$, and any transformation $T$ of length $t$, 
  \[
    (p(1-\epsilon'))^t - 1/n^2 \leq \Pr_{\rho}[T \text{ is a prefix of }\calT(x,y,\rho)] \leq (p(1 + \epsilon'))^t.
  \]
  So long as $(1 \pm \epsilon')^t = \Theta(1)$ we are done.  Clearly this is the case for $\epsilon' = O(1/n)$ since $t \leq 2d = O(n)$.
  We prove each bound in Equations~\eqref{perturbedps1} and~\eqref{perturbedps2} one term at a time for $\epsilon' = O(1/n)$.  

  \noindent
  First inequality (recall that $p_a \leq 1/2$): 
  \begin{align*}
    \frac{\hat{p_a}(1 - \hat{p_a})\hat{p_r}}{(1 - \hat{p_a}^2)} &> \frac{p_a(1 - p_a(1 + 1/n))p_r}{1 - p_a^2}\\
    &= p - \frac{p_ap_r}{n(1 - p_a^2)} \\
    &= p - \frac{p}{n(1 - p_a)} \\
    &= p(1 - O(n))
    \end{align*}
  Second inequality: 
  \begin{align*}
    \frac{\hat{p_a}(1 - \hat{p_a})\hat{p_r}}{(1 - \hat{p_a}^2)} &< \frac{p_a(1 + 1/n)^2(1 - p_a)p_r}{1 - (p_a(1 + 1/n))^2} = \frac{p_a(1 - p_a)p_r}{1/(1 + 1/n)^2 - p_a^2} \\
    &= \frac{p_a(1 - p_a)p_r}{1 - O(1/n^2) - p_a^2} < \frac{p_a(1 - p_a)p_r}{(1 - p_a^2) (1 - O(1/n^2))} \\
    &= p(1 + O(1/n^2))
  \end{align*}
  Third inequality (since $p_a \leq 1/2$, $2p_a \leq 4(1 - p_a)^2$): 
  \begin{align*}
    \frac{(1-\hat{p_a})^2\hat{p_r}^2}{(1 - \hat{p_a}^2)} &\geq 
    \frac{(1-p_a(1 + 1/n))^2p_r^2}{(1 - {p_a}^2)} \\
    &= \frac{(1 - 2p_a(1 + 1/n) + p_a^2(1 + 1/n)^2)p_r^2}{(1 - {p_a}^2)} \\
    &> \frac{((1-p_a)^2 - 2p_a/n)p_r^2}{(1 - {p_a}^2)} = p - \frac{2p_ap_r^2}{n(1 - {p_a}^2)} \\
    &\geq p(1 - O(1/n)) 
  \end{align*}
  Fourth inequality (largely the same as the second inequality): 
  \begin{align*}
    \frac{(1-\hat{p_a})^2\hat{p_r}^2}{(1 - \hat{p_a}^2)} &\leq 
    \frac{(1-{p_a})^2p_r(1 + 1/n)^2}{1 - ({p_a}(1 + 1/n))^2} = p(1 + O(1/n^2))
  \end{align*}
\end{proof}




%
%


\bibliographystyle{plain}
\bibliography{main}

\appendix
\section{Pseudocode and Example Hash}
\seclabel{example}

Below we give an example of how three strings $x$, $y$, and $z$ are hashed using an underlying function $\rho_1$.  We use $\Sigma = \{a,b,c\}$ and $p = 1/8$, so $p_a = 1/3$ and $p_r = 1/2$.  For simplicity, we round the values of $\rho_1$ to the first decimal place, and truncate the domain of $\rho_1$ to only show values of $|s|$ up to $5$.


%
%
{
  \centering
\vspace{.4in}
$x = abc$\\
$h_{\rho_1}(x) = \bot a \bot \bot \bot \bot$\\
$y = bac$\\
$h_{\rho_1}(y) = \bot a \bot \bot \bot \bot$\\
$z = cba$\\
$h_{\rho_1}(z) = c \bot  \bot a \$ $\\
}

  \vspace{.4in}
  {
    \centering
\begin{tikzpicture}
    \foreach \i in {0,...,3}
       \foreach \j in {0,...,3}
          \node[circle,thick,draw] (\i\j) at (\i,3-\j) {};

    \node[label=270:\die,circle,thick,draw] (stop) at (-.8,-1.2) {};

    \node (x1) at (0,3.7) {$a$};
    \node (x2) at (1,3.75) {$b$};
    \node (x3) at (2,3.7) {$c$};
    \node (x4) at (3,3.7) {$\$$};

    \node (y1) at (-.5,3) {$b$};
    \node (y2) at (-.5,2) {$a$};
    \node (y3) at (-.5,1) {$c$};
    \node (y4) at (-.5,0) {$\$$};

    \foreach \i in {0,...,2}
       \foreach \j in {0,...,2}
          \path[->,draw,thick,auto] (\i\j) edge node {} (\the\numexpr\i+1\relax\the\numexpr\j+1\relax);

    \path[blue,->,draw,line width=2] (01) edge node {} (12);
    \path[blue,->,draw,line width=2] (22) edge node {} (33);

    \foreach \i in {0,...,3}
       \foreach \j in {0,...,3}
          \path[->,draw,thick,distance=4mm,in=30,out=70] (\i\j) edge node {} (\i\j);

    \path[->,draw,thick,distance=4mm,in=30,out=70] (stop) edge node {} (stop);

    \path[blue,->,draw,line width=2,distance=4mm,in=30,out=70] (22) edge node {} (22);

       \path[blue,->,draw,line width=2] (00) edge node {} (01);
       \path[->,draw,thick] (11) edge node {} (12);
       \path[->,draw,thick] (21) edge node {} (22);
       \path[->,draw,thick] (02) edge node {} (03);
       \path[->,draw,thick] (12) edge node {} (13);
       \path[->,draw,thick] (30) edge node {} (31);
       \path[->,draw,thick] (31) edge node {} (32);
       \path[->,draw,thick] (32) edge node {} (33);

       \path[->,draw,thick] (00) edge node {} (10);
       \path[->,draw,thick] (11) edge node {} (21);
       \path[->,draw,thick] (21) edge node {} (31);
       \path[->,draw,thick] (02) edge node {} (12);
       \path[blue,->,draw,line width=2] (12) edge node {} (22);
       \path[->,draw,thick] (03) edge node {} (13);
       \path[->,draw,thick] (13) edge node {} (23);
       \path[->,draw,thick] (23) edge node {} (33);

       \path[->,draw,thick,opacity=.29] (00) edge [bend right] node {} (stop);
       \path[->,draw,thick,opacity=.29] (11) edge [bend right] node {} (stop);
       \path[->,draw,thick,opacity=.29] (21) edge [bend right] node {} (stop);
       \path[->,draw,thick,opacity=.29] (02) edge [bend right] node {} (stop);
       \path[->,draw,thick,opacity=.29] (22) edge [bend left] node {} (stop);
       \path[->,draw,thick,opacity=.29] (30) edge [bend left] node {} (stop);
       \path[->,draw,thick,opacity=.29] (31) edge [bend left] node {} (stop);
       \path[->,draw,thick,opacity=.29] (32) edge [bend left] node {} (stop);
       \path[->,draw,thick,opacity=.29] (03) edge [bend left] node {} (stop);
       \path[->,draw,thick,opacity=.29] (13) edge [bend left] node {} (stop);
       \path[->,draw,thick,opacity=.29] (23) edge [bend left] node {} (stop);

\end{tikzpicture}
    
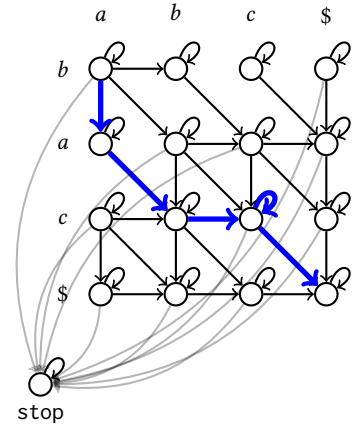
\captionof{figure}{This figure shows $G(x,y)$.  For clarity, all edge labels are ommited and \die{} edges are partially transparent. The edges traversed by $g(x,y,\rho_1)$ are bold and colored blue.}
  }
    \begin{center}
\begin{tabular}{|c|c|c|}
\hline
  ${x_i}$ & ${|s|}$ & ${\rho_1(x_i,|s|)}$ \\
\hline
\hline
  a & 0 & $(0.1, 0.7)$\\
\hline
  b & 0 & $(0.6, 0.3)$\\
\hline
  c & 0 & $(0.7, 0.6)$\\
\hline
  \$ & 0 &$(0.1, 0.4)$\\
\hline
  a & 1 & $(0.9, 0.6)$\\
\hline
  b & 1 & $(0.8, 0.3)$\\
\hline
  c & 1 & $(0.5, 0.9)$\\
\hline
  \$ & 1 &$ (0, 0.1)$\\
\hline
  a & 2 & $(0.1, 0.7)$\\
\hline
  b & 2 & $(0.8, 0.2)$\\
\hline
  c & 2 & $(0.1, 0.9)$\\
\hline
  \$ & 2 &$(0.1, 0.3)$\\
\hline
  a & 3 & $(0.6, 0.8)$\\
\hline
  b & 3 & $(0.9, 0.4)$\\
\hline
  c & 3 & $(0.2, 0.8)$\\
\hline
  \$ & 3 &$(0.8, 0.7)$\\
\hline
  a & 4 & $(0.2, 0.3)$\\
\hline
  b & 4 & $(0.1, 0.1)$\\
\hline
  c & 4 & $(0.7, 0.4)$\\
\hline
  \$ & 4 &$(0.9, 0.5)$\\
\hline
  a & 5 & $(0.5, 0.6)$\\
\hline
  b & 5 & $(0.1, 0.5)$\\
\hline
  c & 5 & $(0.4, 0.6)$\\
\hline
  \$ & 5 & $(0.6, 0)$\\
\hline
\end{tabular}
    \end{center}

  \begin{algorithm}[H]
    \caption{Calculating $h_{\rho}(x)$}
  \begin{algorithmic}[1]
    \State{$i \gets 0$}
    \State{Create an empty string $s$}
    \While{$i < |x|$ and $|s| < 8d/(1-p_a) + 6\log n$}
        \State{$(r_1,r_2)\gets \rho(x_i,|s|)$}
        \If{$r_1 \leq p_a$}
            \State{Append $\bot$ to $s$}
        \ElsIf{$r_2 \leq p_r$}
            \State{Append $\bot$ to $s$}
            \State{$i \gets i + 1$}
        \Else
            \State{Append $x_i$ to $s$}
            \State{$i \gets i + 1$}
        \EndIf
    \EndWhile
    \State{\Return{$s$}}
\end{algorithmic}
\end{algorithm}

\end{document}